\documentclass[10pt, conference]{IEEEtran}  
\IEEEoverridecommandlockouts   
\usepackage{amsthm,amsmath,amssymb,dsfont,array,mathrsfs,graphicx}

\newcommand{\minimize}{\operatornamewithlimits{minimize}}

\newtheorem{theorem}{Theorem}[section]
\newtheorem{prop}[theorem]{Proposition}
\newtheorem{lemma}[theorem]{Lemma}

\newcommand{\post}[2]{\begin{center} \includegraphics[width=#2cm]{#1} \end{center} }

\newcommand{\bs}[1]{\boldsymbol #1}													
\newcommand{\mc}[1]{\ensuremath{\mathcal{#1}}}	
\newcommand{\mb}[1]{\ensuremath{\mathbf{#1}}}
\newcommand{\beq}[1]{\begin{equation} \label{eq:#1}}
\newcommand{\eeq}{\end{equation}}
\newcommand{\beqn}{\begin{equation*}}
\newcommand{\eeqn}{\end{equation*}}

\newcommand{\indicator}[1]{\textbf{1}_{\left\{#1\right\}}} 	

\newcommand{\wh}[1]{\widehat{#1}}														
\newcommand{\wt}[1]{\widetilde{#1}}

\begin{document}
%
\title{On the Incentive to Deviate in Core Selecting \\Combinatorial Auctions
 }

\author{\IEEEauthorblockN{Vineet Abhishek and Bruce Hajek
\thanks{This work was supported by the National Science Foundation
under Grant NSF ECCS 10-28464.}}\\
\IEEEauthorblockA{Department of Electrical and Computer Engineering\\
and the Coordinated Science Laboratory  \\
University of Illinois, 
Urbana, IL,  USA\\
Email: \{abhishe1,b-hajek\}@illinois.edu}
}
\maketitle
\date{\today}

\begin{abstract}
Recent spectrum auctions in the United Kingdom, and some proposals
for future auctions of spectrum in the United States, are based on
preliminary price discovery rounds,  followed by calculation of
final prices for the winning buyers.   For example, the prices could
be the projection of Vikrey prices onto the core of reported prices.
The use of Vikrey prices should lead to more straightforward bidding,
but the projection reverses some of the incentive for bidders to report
truthfully.  Still, we conjecture that the price paid by a winning
buyer increases no faster than the bid, as in a first price auction. 
It would be rather disturbing if the conjecture is false.   The
conjecture is established for a buyer interacting with disjoint groups of
other buyers in a star network setting.   It is also shown that for any
core-selecting payment rule and any integer w greater than or equal to
two, there is a market setting with w winning buyers such that the
price paid by some winning buyer increases at least (1-1/w) times
as fast as the price bid.
\end{abstract}

\section{Introduction} \label{sec:intro}
A practical auction design must satisfy a variety of constraints resulting from legal or political policies; the resulting allocation and payments must be \textit{fair} in some sense. However, the central approach in auction theory has been to first enforce truthful reporting of the information privately held by the agents and then try to achieve the objectives of revenue or social welfare maximization. Consequently, the payment rules suggested by some classical auction mechanisms are deemed unfair. For example, for a Vikrey auction under complementarity, the winning buyers may end up paying a price lower than the offers made by the losing buyers \cite{AusubelMilgrom06}. \textit{Core-selecting auctions} (henceforth, CSAs) are a way to alleviate problems such as unfair pricing or low revenue \cite{DayMilgrom08}. In a CSA, there is no group of losing buyers whose original bids could be combined with the payments of some subset of the winning buyers to increase the revenue of the seller. Such outcome is competitive and is said to be in \textit{core}. CSAs have been taken seriously by policymakers, e.g., in the recent spectrum auctions in the UK \cite{Cramton12}.

In general, it is impossible to design an auction that satisfies both truthful reporting of the preferences and the outcome in core \cite{GoereeLien09}. In a CSA, the outcome is in core with respect to the reported preferences and not necessarily with respect to the true preferences. Furthermore, a core outcome does not uniquely determine the payments of the winning buyers.  Hence, a core selecting auction must specify a rule for deciding the payments of the winning buyers; see \cite{DayCramton11}, \cite{ErdilKlemperer10}, and \cite{AusubelBaranov10} for some simple core-selecting payment rules. The choice of the payment rule affects the bidding strategy and the incentives for truthful reporting of the preferences (see, e.g., \cite{AusubelBaranov10}).

This paper analyzes the incentive to deviate from truthful bidding for CSAs. The metric we use is the \textit{marginal incentive to deviate} (henceforth, MID). The MID measures how much a buyer's payment increases per unit increase in his bid for a small increase in the bid, keeping the bids of others constant. This metric is inspired by \cite{ErdilKlemperer10}; as argued there, buyers are unlikely to find the best possible deviation from truthful bidding, but may have a clearer view of where and how to gain from smaller deviations. We, however, consider the MID for each buyer rather than the sum of MIDs across the buyers. Moreover, the analysis in \cite{ErdilKlemperer10} is for a simple example; we consider a benchmark class of combinatorial auctions with a \textit{single-parameter} environment and a \textit{star-network} setting (to be made precise in Section \ref{sec:star}). We treat MID smaller than one as a basic requirement for any reasonable payment rule; however, there are no results in the existing literature on the same. Our particular focus is on the payment rule that minimizes the Euclidean distance between the Vikrey price vector and the set of core payment vectors (referred to as the \textit{quadratic payment rule}); and a variant of it where among the core payment vectors with the minimum sum (referred to as the \textit{minimum revenue core} (MRC)), the one closest to the Vikrey payment vector is selected (see \cite{DayCramton11} for further details). Our main results are:
\begin{enumerate}
\item
We show for the star network setting that the MID for the quadratic payment rule and its MRC variant is at most one. We conjecture that the MID smaller than one is true in general for the quadratic payment rule. Our results provide strong theoretical evidence in support of this conjecture.
\item
We show that the worst case MID for any core-selecting payment rule over all environments when there are $w$ winners
is at least $1-\frac{1}{w}.$    This quantifies the loss in the incentives for truthful bidding if core-selecting outcome is imposed as a constraint.
\end{enumerate}

The  rest of the paper is organized as follows. Section \ref{sec:model} gives notation and preliminaries. Section \ref{sec:star} describes the star network setting and presents the bounds on the MID. Section \ref{sec:mid-general} presents a lower bound on the worst-case
MID for a general core-selecting payment rule.  We conclude in Section \ref{sec:conclusions} with some directions for future research.

\section{Model and Preliminaries} \label{sec:model}
Consider an auction with multiple buyers, multiple items to be sold, and a single seller.  
Let  $M$ denote the set of items to be sold and $N$ denote the set of buyers. We limit discussion in this paper to the case of single-parameter buyers, meaning that each buyer $j$ submits only a single bid $b_j$ for a particular bundle $S_j,$  with $S_j \subseteq M.$ A given buyer $j$ will either be allocated the set $S_j$ of items, or will be allocated no items.

A {\em coalition} of buyers is simply a subset $T$ of $M.$ A coalition $T$ is called {\em feasible} if the bundles in $T$ are mutually disjoint; i.e., $S_j \cap S_k = \emptyset$ for all $i, j \in T$ such that $i \neq j$. Let $\mc{T}$ denote the set of feasible coalitions.

The outcome of the auction is a pair $(W,\mb{p})$. Here, $W$ is the set of winners (hence, a feasible coalition) and $\mb{p} \triangleq (p_j : j \in W)$ is a vector of prices; buyer $j$ pays the seller $p_j$ for the bundle of items $S_j.$ The outcome is said to be \textit{efficient} if the winning coalition $W \in \mc{T}$ is a solution of the following {\em winner determination problem}:
\beq{eff-alloc}
\sum_{j \in W} b_j = \max_{T \in \mc{T}}  \sum_{j\in T}  b_j.
\eeq

The Vikrey price vector $\mb{v}$ for an efficient outcome $W$ is the vector $\mb{v} \triangleq (v_j : j \in W)$, defined by:
\beq{vcg-price}
v_j  = \left\{\max_{T \in \mc{T}: j \not\in T}  \sum_{i\in T}  b_i \right\}  - \sum_{i\in W-j}  b_i. 
\eeq
In rest of this paper, the Vikrey auction is denoted by the outcome $(W,\mb{v}).$

The following example illustrates the motivation behind core selecting auctions. There are two small buyers, one big buyer, and two items. Each small buyer wants a single item and has value $\$8$. The big buyer wants both items and has value $\$10$. Under the Vikrey price vector,
the two small buyers win an item each at price $\$2$. The seller's total revenue is $\$4$ which is smaller than the bid of the big buyer. The big buyer might consider this outcome unfair.

An outcome $(W,\mb{p})$ is said to be {\em blocked} by a coalition $C$ if
\beqn
 \sum_{i\in W} p_i   <   \sum_{i\in C}  (b_i \indicator{ i\not\in W }  + p_i  \indicator{ i \in W}) ,
\eeqn
which means that the seller could raise more revenue by switching the set of winners from $W$ to $C,$
and charging $b_i$ to those buyers in $C/W,$  and keeping the prices equal to $p_i$ for buyers $i$ in $C\cap W.$ Equivalently, the outcome $(W,\mb{p})$ is blocked by a coalition $C$ if
$
\sum_{i\in W / C} p_i   <    \sum_{j \in C/W} b_j.
$
An outcome $(W,\mb{p})$ is said to satisfy the coalition core constraints if there are no blocking coalitions,
or equivalently, if
$
\sum_{i\in W / C} p_i   \geq     \sum_{j\in C/W} b_j 
$
for all feasible coalitions $C.$  Note that if $C$ has the form $C=W-j$ for some $j\in W,$  then the
constraint for $C$ becomes $p_j \geq 0.$ The overall core region is defined by the coalition core constraints
 and \textit{individual rationality} (IR) constraints, $p_i \leq b_i$ for all $i \in W$.  
 
The coalition core constraints can be written as $\mb{Ap} \geq \bs{\beta}.$ Here, $\mb{A}$ is the matrix such that for each coalition $C$, there is a row of $\mb{A}$ that is the binary indicator vector for the set $W / C$; the corresponding coordinate of the column vector $\bs{\beta}$ is:
$
\beta_C = \sum_{i\in C / W} b_i.
$
In summary, the core region for the reported bid vector $\mb{b}$ and winner set $W$ is the set of price
vectors $\mb{p}$ satisfying $\mb{Ap} \geq \bs{\beta}$ and $\mb{p} \leq \mb{b}.$ The core region includes the vector formed by the bids of the winning buyers.

For a price vector $\mb{r},$ the quadratic rule for payment determination \cite{DayCramton11} with reference price vector $\mb{r}$ is: 
\begin{center}
\parbox{3.5in}{
{\em  QUADRATIC($\mb{A},\mb{b},\mb{r},\bs{\beta}$):}
\beqn
\minimize_{\mb{p}} (\mb{p}-\mb{r})^T(\mb{p}-\mb{r}),
\eeqn
subject to:  $\mb{Ap} \geq \bs{\beta}$ and 
$\mb{p} \leq \mb{b}.$}
\end{center} \vspace{0.05in}

We focus on the auction mechanism that uses the quadratic rule for payment with the Vikrey price vector as the reference price vector. The Lagrangian function for the optimization problem QUADRATIC with the Vikrey price vector as the reference price vector is given by:
\begin{eqnarray*}
\lefteqn{ L(\mb{p},\bs{\lambda}, \bs{\mu}) }&& \\
&=& (\mb{p}-\mb{v})^T(\mb{p}-\mb{v})  + \bs{\lambda}^T (\bs{\beta} -\mb{Ap}) + \bs{\mu}^T (\mb{p}-\mb{b}),
\end{eqnarray*}
where $\bs{\lambda}$ is the vector of Lagrange multipliers for the constraint $\mb{Ap} \geq \bs{\beta}$ and
$\bs{\mu}$ is the vector of Lagrange multipliers for the constraint $\mb{p}\leq \mb{b}.$ Since the objective function in QUADRATIC is strictly convex and the constraints are linear, there exists a unique solution and
there exists corresponding values of the Lagrange multipliers satisfying the Karush-Kuhn-Tucker conditions.
The result is that $\mb{p}$ is the solution of QUADRATIC($\mb{A},\mb{b},\mb{v},\bs{\beta}$) if and only if there exist values of the vectors $\bs{\lambda}$ and $\bs{\mu}$ satisfying:
\begin{eqnarray*}
\mb{Ap} & \geq & \bs{\beta},  \\
\mb{p}  &\leq & \mb{b},  \\
\bs{\lambda}  &\geq & 0,  \\
\bs{\mu} & \geq  & 0, \\
\mb{p} & =  & \mb{v} + \mb{A}^T \bs{\lambda} - \bs{\mu},  \\
\bs{\lambda}^T(\bs{\beta}-\mb{Ap}) & = &  0,  \\
\bs{\mu}^T (\mb{p}-\mb{b}) & = & 0.
\end{eqnarray*}


\section{Star Network Setting} \label{sec:star}
This section describes the star network setting that we consider and obtains bounds on the MID.

\subsection{Quadratic payment rule and the MID} \label{sec:qpr}
Suppose there are $1+\sum_{j=1}^J n_j$  items,  for some positive integers $J, n_1, \ldots , n_J,$
labeled as follows. There is an item labeled zero, and for $1 \leq j \leq J$ and $1 \leq k \leq n_j$ there is an item labeled  $(j,k).$   Define a buyer to be a single-item buyer if the bundle the buyer is bidding for contains exactly one item. Suppose there are  $2(1+\sum_{j=1}^J n_j)+J$ buyers, with corresponding bundles described as follows. For each item there are two single-item buyers with desired bundle consisting of exactly that item. Such single-item buyers comprise $2(1+\sum_{j=1}^J n_j)$ of the buyers. In addition, there are $J$ buyers, indexed by $1 \leq j \leq J,$ such that $S_j = \{ 0 , (j,1), \ldots  , (j,n_j)\}.$
 
Suppose the set of winning buyers $W$ includes exactly one of the single-item buyers for each of the items.   Thus, there are $1+\sum_{j=1}^J n_j$ winning buyers and $1+J+\sum_{j=1}^J n_j$ losing buyers. We focus on the winning buyers (i.e. those in $W$) and refer to them by the item that they bid for -- thus, the winning buyers consist of buyer zero and buyers $(j,k)$ for $1\leq j \leq J$ and $1\leq k \leq n_j.$

Suppose  values of the following variables are fixed:
\begin{itemize}
\item  $b_{j,k}$ denotes the bid of (winning) buyer $(j,k).$ 
\item  $\underline{b}_{j,k}$ denotes the bid of the losing  single-item buyer that bid for item $(j,k).$
\item $\triangle_{j,k}=b_{j,k} -  \underline{b}_{j,k}.$   Since $W$ is assumed to be a solution of the
winner determination problem,  $\triangle_{j,k}\geq 0$ for $1\leq j \leq J$ and $1\leq k \leq n_j.$
\item   $\underline{b}_{0}$ denotes the bid of the losing single-item buyer that bid for item zero.
\item   $C_j$ denotes the (losing) bid of  buyer $j,$  for package $S_j=\{0, (j,1), \ldots   , (j,n_j) \}.$
\end{itemize}

The values of the following variables are thus also determined:
\begin{itemize} 
\item $v_0 = \max \{  \underline{b}_0 , \max_j \{  C_j - \sum_{k=1}^{n_j}  b_{j,k}  \}  \}.$   This is the minimum
value that buyer zero (i.e. the winning single-item buyer bidding for item zero) must bid in order to be a
winning buyer, as assumed.   It is also the Vikrey price of buyer zero.
\item $\eta_j = v_0 +\left( \sum_{k=1}^{n_j}  b_{j,k} \right) -  C_j  $   for $1\leq j \leq J.$   The fixed value
$\eta_j$ has the following interpretation.   If buyer zero were to submit the minimum possible winning bid,
$v_0,$  then $\eta_j$ is the minimum amount buyer $j$  would need to increase her
bid in order to be eligible for winning.  Note that $\eta_j \geq 0$ for $1\leq j \leq J$ by the choice of $v_0.$
\end{itemize}

The above variables do not include the bid, $b_0,$ of buyer zero (the winning single-item buyer that bids for
item zero) because we allow $b_0$ to vary in the range $[v_0, +\infty).$ For convenience,
we parameterize $b_0$ by  $b_0 = v_0 + \theta,$ and we consider values of $\theta$ with $\theta \geq 0.$
Note that the Vikrey price for buyer zero, $v_0,$  does not depend on $\theta.$ Let $p_{0,\theta}$ denote the price paid by buyer zero, for the price vector determined by the QUARDATIC rule for payment determination, with the Vikrey price vector used as the reference point.

Our main result is Proposition \ref{prop.marg_core} below; it shows that for the star network setting we consider, the MID for the quadratic payment rule is at most one. We treat MID smaller than one as a basic requirement for any reasonable payment rule. In this sense, we establish that the quadratic payment rule passes this basic check.

\begin{prop} \label{prop.marg_core}
The price for buyer zero,  $p_{0,\theta},$  is piecewise linear in $\theta$ with slope less than or equal to one for all $\theta \geq 0.$
\end{prop}

The proof follows from a sequence of lemmas. We start with some notation. Let $v_{j,k,\theta}$ denote the Vikrey price of buyer $(j,k).$  It is given by:
\beqn
v_{j,k,\theta}= \max\{b_{j,k} - \eta_j - \theta , \underline{b}_{j,k} \}.
\eeqn
Equivalently,
\begin{equation} \label{eq.Vikrey_drop}
b_{j,k}-v_{j,k,\theta} = \min\{ \eta_j + \theta,  \triangle_{j,k}\}. 
\end{equation}
In words, \eqref{eq.Vikrey_drop} tells us that the price reduction in going from the bid of buyer $(j,k)$ to
the Vikrey price of buyer  $(j,k)$ is the minimum of  $\eta_j+\theta$   (insures that buyer $j$, who bid
for $S_j$, would still lose if buyer $(j,k)$ switched her bid to $v_{j,k}$)  and  $ \triangle_{j,k}$ 
(insures that the buyer bidding  $\underline{b}_{j,k}$ for item $(j,k)$ indeed loses.) The entire Vikrey price vector can thus be denoted by:
$\mb{v}_{\theta} \triangleq (v_0,[v_{j,k,\theta}]_{1\leq j \leq J, 1\leq k \leq n_j}).$

To ensure there is no blocking coalition $C$, it suffices to consider two types of coalitions $C$, one
for which the buyers bidding for individual items in $S_j$ are replaced by buyer $j$ bidding for bundle $S_j$, and one for
which a single winning buyer is replaced by the losing buyer bidding for the same item.  Therefore,
in the notation we have introduced, the core region for the star network setting is the
set of vectors $\mb{p}$ satisfying
\begin{eqnarray}
p_0 +\sum_k  p_{j,k}  \geq C_j   \label{eq.sum_cons} \\
p_{j,k}  \in [\underline{b}_{j,k}, b_{j,k} ]   \label{eq.pij_constraints} \\
p_0 \in [v_0,   v_0+ \theta ]~~  \label{eq.p0_constraint}
\end{eqnarray}
We remark that the constraints \eqref{eq.sum_cons} and \eqref{eq.pij_constraints}
imply that   $p_0 \geq  \max_j \{  C_j - \sum_{k=1}^{n_j}  b_{j,k}\},$  so 
we can use \eqref{eq.p0_constraint} instead of $p_0 \in [\underline{b}_0,v_0+\theta]$
in the description of the core region.

One of the key ideas of the proof is to consider a relaxation of the optimization problem
QUADRATIC.   Specifically, we drop
the IR constraint for buyer zero (i.e. the requirement $p_{0,\theta} \leq v_0+\theta$), while retaining
all other constraints, to obtain the
{\em expanded core} region, which is the set of $\mb{p}$ satisfying
\begin{eqnarray}
p_0 +\sum_k  p_{j,k}  \geq C_j   \label{eq.sum_con_ex} \\
p_{j,k}  \in [\underline{b}_{j,k}, b_{j,k} ]   \label{eq.pij_constraints_ex}  \\
p_0  \geq v_0   \label{eq.p0_constraint_ex}
\end{eqnarray}
Note that the expanded core region does not depend on $\theta.$
The relaxation of the optimization problem QUADRATIC that we consider is to project
the Vikrey vector $v_{\theta}$ onto the expanded core region.
Let $\wt{\mb{p}}_\theta \triangleq (\wt{p}_{0,\theta}, [\wt{p}_{j,k,\theta}]_{1\leq j \leq J, 1\leq k \leq n_j }),$ denote the resulting price vector.
That is, let $\wt{\mb{p}}_\theta$ minimize $ (\mb{p}-\mb{v}_{\theta})^T(\mb{p}-\mb{v}_{\theta})$
over $\mb{p}$ subject to \eqref{eq.sum_con_ex} - \eqref{eq.p0_constraint_ex}.

Note that, on one hand, if $\wt{p}_{0,\theta} \leq v_0+\theta,$  then the IR constraint on buyer 0 is not active
for the determination of $\mb{p}_{\theta}.$   On the other hand, if $\wt{p}_{0,\theta} > v_0+\theta,$  then the IR
constraint on buyer 0 is active for the determination of $\mb{p}_{\theta},$ in which case $p_{0,\theta}=v_0+\theta.$
Therefore,  in general, $p_{0,\theta} = \min\{v_0+\theta, \wt{p}_{0,\theta}\}.$
So to prove Proposition \ref{prop.marg_core} it suffices to prove that $\wt{p}_{0,\theta}$ is piecewise linear, and, over
 the set of $\theta$ such that $\widetilde{p}_{0,\theta} \leq v_0 + \theta,$ its slope is less than or equal to one.
The remainder of this proof deals with the relaxed optimization problem.  
The relaxation leads to new values of the Lagrange multipliers, but for brevity
and by abuse of notation, we omit tilde's over the Lagrange multipliers.

If for some $j$,   $\triangle_{j,k}=0$ for all $k$, then the prices of the buyers $(j,1), \ldots  , (j,n_j)$
are frozen at the prices bid, and, after a possible adjustment to $\underline{b}_0,$ the buyers$(j,1), \ldots  , (j,n_j)$ could
be removed.   To avoid trivialities, we can therefore assume without loss of generality that
$\max_{k} \triangle_{j,k} >0$ for each $j.$  Such assumption is in force for the remainder of this section.

\begin{lemma}  \label{lemma.star_net_optimality_cond}  The following hold:\\
(a) 
For any $\theta \geq 0,$  the price vector $\widetilde{\mb{p}}_{\theta}$ and corresponding vector
of Lagrange multipliers $\boldsymbol{\lambda}_{\theta}\stackrel{\triangle}{=}
(\lambda_{1,\theta}, \lambda_{2,\theta}, \ldots  , \lambda_{J,\theta})$  are uniquely determined by the following conditions (where $j$ ranges over $1\leq j \leq J,$  and for each $j,$
$k$ ranges over $1 \leq k \leq n_j)$:
\begin{eqnarray}
\lambda_{j,\theta} & \geq & 0,
\label{eq.cond1}  \\
\widetilde{p}_{0,\theta} & = &  v_0 +  \sum_{j=1}^J  \lambda_{j,\theta},  \label{eq.cond2}    \\
\widetilde{p}_{j,k,\theta} & = &  \min\{v_{j,k,\theta}+ \lambda_{j,\theta},  b_{j,k}\}, 
\label{eq.cond3}
\end{eqnarray}
\begin{equation}
 \sum_{i=1}^J \lambda_{i,\theta} + \eta_j 
  \geq  \sum_{k=1}^{n_j} \max\{ \min\{\eta_j + \theta  , \triangle_{j,k} \} - \lambda_{j,\theta} , 0\}, 
 \label{eq.cond4}
 \end{equation}
 \begin{equation}
 \mbox{equality holds in \eqref{eq.cond4} for $j$ with}~  \lambda_{j,\theta} > 0  \label{eq.cond5}
\end{equation}  

\noindent
(b)  For each $j,$  $\lambda_{j,\theta} < \max_k \triangle_{j,k}$ for all $\theta \geq 0.$

\noindent
(c)   If $\theta > 0$, then $\eta_j + \theta > \lambda_{j,\theta}$  for $1\leq j \leq J.$ 

\noindent
(d)  The variables $\lambda_{j,\theta}$ for each $j,$   and $\sigma_{\theta},$  are piecewise linear functions of $\theta.$

\noindent
(e)  The variable $\sigma_{\theta}$ is nondecreasing in $\theta.$ 

\end{lemma}
\begin{proof}
To begin we first check that \eqref{eq.cond1}-\eqref{eq.cond5} are translations of the KKT conditions for the optimization problem
defining $\wt{\mb{p}}_\theta.$   
For each $j,$   $\lambda_{j,\theta}$ is the Lagrange multiplier for the constraint  \eqref{eq.sum_con_ex}. 
For each $(j,k)$, since $v_{j,k,\theta} \geq \underline{b}_{j,k},$ the constraint $p_{j,k,\theta} \geq \underline{b}_{j,k}$
is not active.   Similarly, the constraint $p_0 \geq v_0$ is not active.   So no Lagrange multipliers are introduced for those
constraints.    We could introduce a Lagrange multiplier $\mu_{j,k}$ for each constraint $p_{j,k} \leq b_{j,k}$, and
then $\wt{p}_{j,k,\theta}=v_{j,k,\theta}+ \lambda_{j,\theta} - \mu_{j,k},$   but since this is a scalar constraint, the variable
$\mu_{j,k}$ can be eliminated, resulting in \eqref{eq.cond3}.  Since  $\wt{\mb{p}}_\theta$ satisfies the
constraint \eqref{eq.sum_con_ex},
\begin{equation} \label{eq.constraintj}
\widetilde{p}_{0,\theta}  + \sum_{k=1}^{n_j}  \widetilde{p}_{j,k,\theta}  \geq C_j,
\end{equation}
which must hold with equality if $\lambda_{j,\theta} > 0.$
 Condition \eqref{eq.cond3} can be rewritten, using \eqref{eq.Vikrey_drop}, as
\begin{eqnarray}
\lefteqn{b_{j,k} - \widetilde{p}_{j,k,\theta}  } && \nonumber  \\
&= &   \max\{ b_{j,k} -v_{j,k,\theta} -\lambda_{j,\theta}, 0 \} \nonumber \\
 & = & \max\{ \min\{  \eta_j+\theta, \triangle_{j,k}\} -\lambda_{j,\theta}, 0\}   \label{eq.b_minus_p}.
\end{eqnarray}
Equation \eqref{eq.b_minus_p} has a natural interpretation.   The left-hand side is the amount by which the final price (computed using the relaxation) for buyer $(j,k)$ is reduced from the original bid.    It is nominally smaller,  by the amount $\lambda_{j,\theta},$   than the  original  Vikrey price reduction given by \eqref{eq.Vikrey_drop}.  The term $\lambda_{j,\theta}$ indicates how much of the original price reduction buyer $(j,k)$ needs to give back to the seller to help meet the core constraint corresponding to $S_j.$  The amount given back cannot be negative, however, because of the IR constraint for buyer  $(j,k).$

For each $j,$ by \eqref{eq.cond2}, the definition of $\eta_j,$ and the constraint \eqref{eq.constraintj},
\begin{eqnarray}
\lefteqn{ \left( \sum_i \lambda_{i,\theta} \right) + \eta_j \nonumber }\\
 & = & (\widetilde{p}_{0,\theta}  -v_0 )  +  v_0 + b_{j,1} + \cdots + b_{j,n_j} - C_j  \nonumber \\  \\
& \geq & 
\begin{array}{c}
 \widetilde{p}_{0,\theta} - v_0 + v_0 + b_{j,1} + \cdots  + b_{j,n_j}   \\ 
 ~~~~ - \widetilde{p}_{0,\theta} - \widetilde{p}_{j,1,\theta} - \cdots  - p_{j,n_j,\theta} \\
 ~~~\mbox{with equality if}~ \lambda_{j,\theta} > 0  \end{array}
 \nonumber \\
& = & \sum_{k=1}^{n_j } (b_{j,k} - \widetilde{p}_{j.k.\theta})   \label{eq.lambda_eta}
\end{eqnarray}
Substituting \eqref{eq.b_minus_p} into \eqref{eq.lambda_eta} 
 yields \eqref{eq.cond4} and \eqref{eq.cond5}.
 
 Keeping in mind the interpretation of  \eqref{eq.b_minus_p}, there is also a natural interpretation of  \eqref{eq.cond4} and
 \eqref{eq.cond5}.
In order to meet the core constraint imposed by the losing buyer $j$ bidding for $S_j$,  the sum of the price reductions
for buyers $(j,1)$ through $(j,n_j)$ should be no more than the original slack $\eta_j$ in the constraint, plus the amount
that buyer zero bids above her Vikrey price.   Moreover, equality should hold if that core constraint is tight.
This concludes the proof that \eqref{eq.cond1}-\eqref{eq.cond5} are translations of the KKT conditions for the optimization problem
defining $\wt{\mb{p}}_\theta.$

Next, we prove the uniqueness part of  Lemma  \ref{lemma.star_net_optimality_cond}(a).
Let $s \geq 0$ and $l\geq 0$ be independent variables, and consider for fixed $j$ with $1\leq j \leq J$ the condition
\begin{equation}    \label{eq.opt_cond_sigma}
\begin{array}{l}
s  + \eta_j  \geq \sum_{k=1}^{n_j}   \max\{ \min\{\eta_j + \theta  , \triangle_{j,k} \} -  l , 0\} , \\ 
\mbox{with equality if}~ l   > 0.
\end{array}
\end{equation}
If  $s + \eta_j > 0$,  \eqref{eq.opt_cond_sigma} uniquely determines $l.$   If $s +\eta_j = 0$
(i.e. $s=\eta_j=0$) then the set of $l$  satisfying  \eqref{eq.opt_cond_sigma} is the interval
$[ \min \{\theta ,   \max  \{  \triangle_{j,k} : 1\leq k \leq n_j \}\} , + \infty ).$    Therefore,
\eqref{eq.opt_cond_sigma} determines a continuous,
piecewise linear nonincreasing function $\phi_j:[0,\infty)\rightarrow [0,\infty)$ so
that for each $s$, the value $l$ given by $l = \phi_j(s)$ 
satisfies \eqref{eq.opt_cond_sigma}, and no other $l$ satisfies \eqref{eq.opt_cond_sigma}
unless $s=\eta_j=0,$ in which case $l$ is the minimum solution to  \eqref{eq.opt_cond_sigma}.
An illustration of the function $\phi_j$ is shown in Figure \ref{fig.phi_j}.
\begin{figure}[htb] 
\post{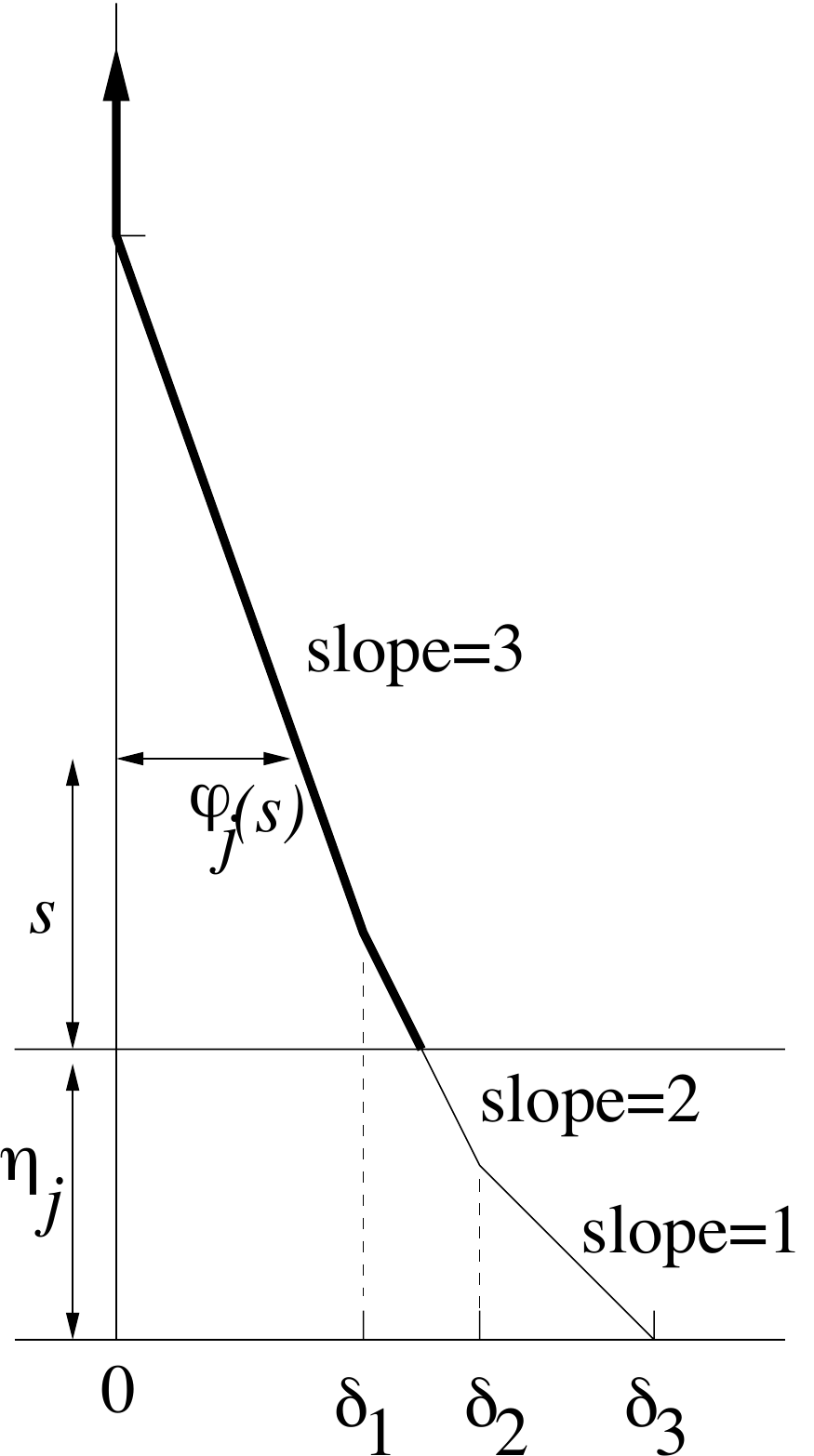}{3}
\caption{Illustration of a function $\phi_j,$  where $\delta_k$ in the figure
is given by $\delta_k=\min\{ \eta_j + \theta, \triangle_{j,k} \}$ for $1\leq k \leq 3.$   \label{fig.phi_j}  }
\end{figure}
Note that for $s \geq 0,$   $ \phi_j(s)=0$ if and only if $s+ \eta_j \geq  \sum_{k=1}^{n_j}  \min\{\eta_j+\theta, \triangle_{j,k} \} .$
so that $\phi_j$ is identically zero if and only if  $\eta_j \geq  \sum_{k=1}^{n_j}  \min\{\eta_j+\theta, \triangle_{j,k} \}.$

Fix any solution $\wt{\mb{p}}_{\theta}$ and $\boldsymbol{\lambda}_{\theta}$  of  \eqref{eq.cond1}-\eqref{eq.cond5},
and let  $\sigma_{\theta} = \sum_{j=1}^J \lambda_{j,\theta}.$    If  $\sigma_{\theta}>0,$
conditions \eqref{eq.cond4}  and \eqref{eq.cond5} require $\lambda_{j,\theta} = \phi_j(\sigma_\theta)$
for each $j,$  yielding that  $\sigma_\theta$ satisfies the following fixed point equation:
\begin{equation}   \label{fixed_pt_sigma}
\sigma_{\theta}  = \sum_{j=1}^J  \phi_j (\sigma_\theta)
\end{equation}
If  $\sigma_{\theta} = 0$ then $\lambda_{j,\theta}=0$ for each $j$ and  \eqref{eq.cond4} yields
that $\eta_j \geq  \sum_{k=1}^{n_j}  \min\{\eta_j+\theta, \triangle_{j,k} \}$ for all $j,$  so that
$\phi_j$ is identically zero for all $j,$ and therefore $\sigma_{\theta}$ again satisfies the fixed point
equation \eqref{fixed_pt_sigma}. 

The right-hand side of \eqref{fixed_pt_sigma} is a piecewise linear nonincreasing function
of $\sigma_\theta,$ so there is a unique solution $\sigma_\theta.$   Thus, $\sigma_{\theta}$ is uniquely
determined by \eqref{eq.cond1}-\eqref{eq.cond5}.   Furthermore, the $\lambda$'s are uniquely determined
because $\lambda_{j,\theta}=\phi_j(\sigma_{\theta}),$  and hence $\wt{p}_{\theta}$ is also uniquely
determined.
This completes the proof of  Lemma  \ref{lemma.star_net_optimality_cond}(a).

Lemma   \ref{lemma.star_net_optimality_cond}(b)  is proved by argument by
contradiction.  If the statement is false for some $j,$   then
$\lambda_{j,\theta} \geq  \max_k \triangle_{j,k}.$  In particular  (using the assumption $\max_k  \triangle_{j,k} > 0$)
$\lambda_{j,\theta}>0$ so equality holds in  \eqref{eq.cond4}.   But also, the right-hand side 
of  \eqref{eq.cond4} is equal to zero, implying that $\lambda_{j,\theta} =0,$ which is
a contradiction.   The proof of Lemma \ref{lemma.star_net_optimality_cond}(b) is complete.

Lemma   \ref{lemma.star_net_optimality_cond}(c)  is also proved by argument by
contradiction.   If   $\eta_{j_o} + \theta \leq \lambda_{j_o,\theta}$ for some $j_o$,  then the right-hand side of \eqref{eq.cond4} for $j=j_o$ is equal to zero. But also $\lambda_{j_o,\theta}> 0$, so equality must hold in \eqref{eq.cond4}, so the left-hand side of  \eqref{eq.cond4} must equal zero for $j=j_o$, which contradicts
the fact that  $\sigma_{\theta} \geq \lambda_{j_o,\theta}> 0.$  Lemma \ref{lemma.star_net_optimality_cond}(c) is proved.

The price vector $\wt{\mb{p}}_{\theta}$ is the projection of $\mb{v}_{\theta}$, which is piecewise linear with respect to
$\theta,$ onto the expanded core region, which is a closed, convex polytope region not depending on $\theta.$   Therefore
$\wt{\mb{p}}_{\theta}$ is piecewise linear with respect to $\theta.$    By parts (b) and (c) already proved, for each $j,$
there is some $k$ so that $\lambda_{j,\theta}< \min\{\eta_j+\theta, \triangle_{j,k}\}$ so that by \eqref{eq.b_minus_p}, $\lambda_{j,\theta}$
can be expressed as a piecewise linear function of $\theta$ and $\widetilde{p}_{j,k,\theta},$ so that $\lambda_{j,\theta}$ is also piecewise
linear in $\theta.$   Since sums of piecewise linear functions are piecewise linear, it
follows that $\sigma_{\theta}$ is also piecewise linear in $\theta.$
Lemma  \ref{lemma.star_net_optimality_cond}(d) is proved.

The function $\phi_j$ is nondecreasing in $\theta$ for each $j,$   which, given that
$\sigma_{\theta}$ is determined by   \eqref{fixed_pt_sigma}, establishes that $\sigma_{\theta}$ is
nondecreasing in $\theta.$   Lemma  \ref{lemma.star_net_optimality_cond}(e) is proved.

 \end{proof}

The next step of the proof of Proposition \ref{prop.marg_core} is to derive an expression for the derivative of $\sigma_{\theta}.$   By
Lemma  \ref{lemma.star_net_optimality_cond}(d), $\sigma_{\theta}$ and $\lambda_{j,\theta}$ for each
$j$ have right-hand derivatives with respect to $\theta.$   We use the notation $\frac{d^+}{d\theta}$ to  denote
right-hand differentiation.
It is simpler and sufficient to find an expression for the derivative of $\sigma_{\theta}$ that holds except for
a finite exceptional set of $\theta$ values, than to identify the right-hand derivative of $\sigma_{\theta}$ everywhere.
For $\theta\geq 0$, let $J(\theta)=\left\{ j : \frac{d^+\lambda_{j,\theta}}{d\theta} \neq 0\right\}.$

\begin{lemma} \label{lemma.Dsigma}  There is a finite set $\cal E$ (described in the proof) so that
for all $\theta \geq 0$ with $\theta \not\in {\cal E}$
\begin{equation}   \label{eq.Dsigma}
  \frac{d^+\sigma_{\theta}}{d\theta}
    =  \frac{      
    \sum_{j \in J(\theta) }   \frac{   |\{ k : \triangle_{j,k} > \eta_j  + \theta \}|      }{     |\{ k : \triangle_{j,k} > \lambda_j  \}|      }
       }{1 + 
      \sum_{j \in J(\theta) }   \frac{  1  }{     |\{ k : \triangle_{j,k} > \lambda_j  \}|      } 
       }.
\end{equation}
\end{lemma}
\begin{proof}
For $1\leq j \leq J,$   define the function $\chi_j$ by:
$$
\chi_j(\theta,\lambda_j) \stackrel{\triangle}{=}
 \sum_{k=1}^{n_j}  \max\{\min\{\eta_j + \theta,\triangle_{j,k} \} - \lambda_j , 0 \}.
$$
A geometric interpretation is that $\chi_j(\theta,\lambda_j)$ is the area of the shaded region
shown in Fig.  \ref{fig.chi_area}.
\begin{figure}[htb] 
\post{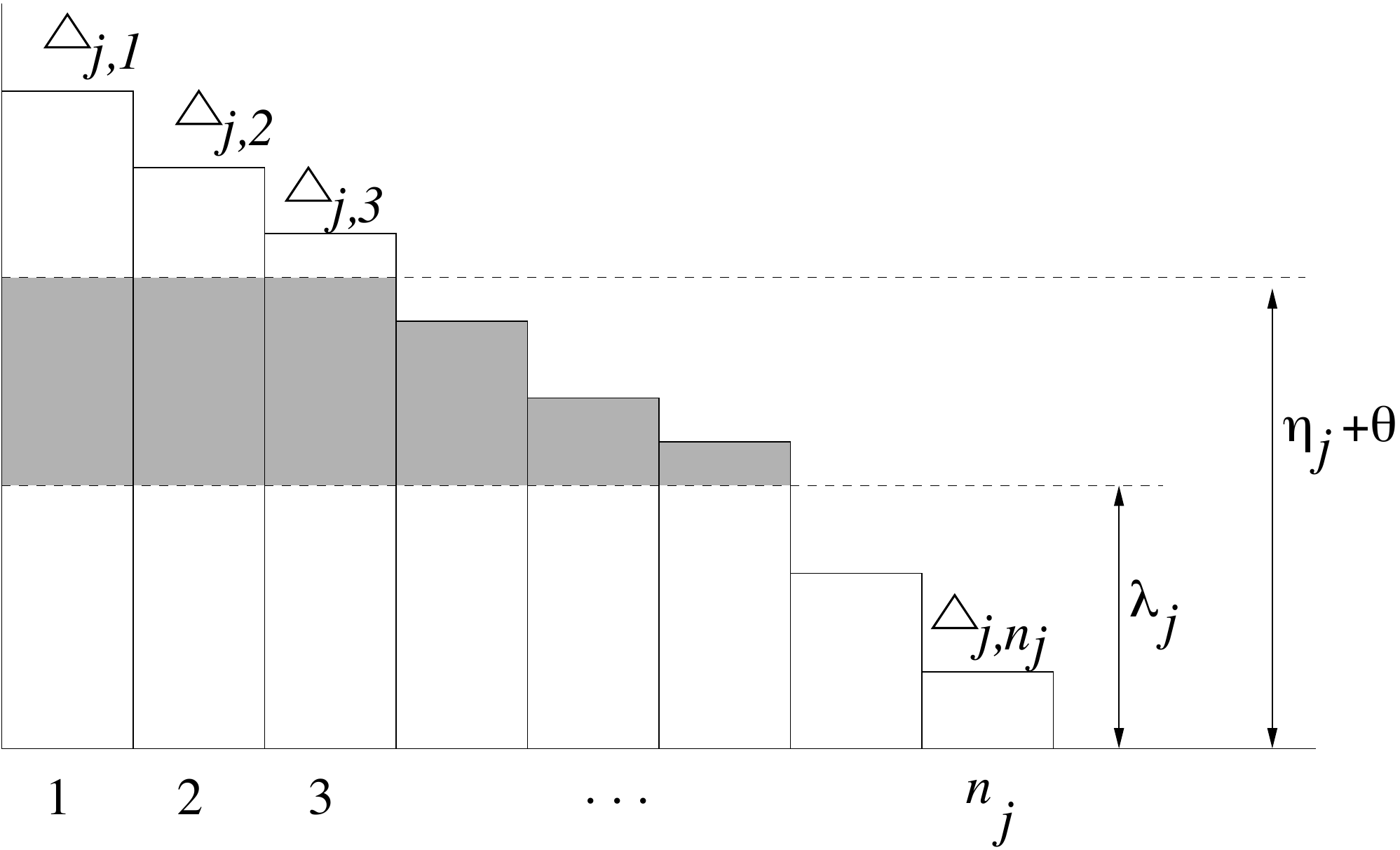}{8}
\caption{The area of the shaded region is $\chi_j(\theta,\lambda_j).$   Figure assumes
ordering:  $\triangle_{j,1} \geq \triangle_{j.2} \geq \cdots  \geq \triangle_{j,n_j}.$   \label{fig.chi_area}  }
\end{figure}
Note that  $\chi_j(\theta,\lambda_j)$ is linear within each connected region of the $(\theta, \lambda_j)$ plane
when said plane is partitioned by the grid
\begin{eqnarray*}
G_j=\{ (\theta, \lambda_j) :  \eta_j + \theta  \in \{ \triangle_{j,1}, \ldots , \triangle_{j,n_j}\}  `\mbox{or}\\~~~~~~ ~ \lambda_j \in 
\{ \triangle_{j,1}, \ldots , \triangle_{j,n_j}\}  \}.
\end{eqnarray*}
 Let $D^+_\theta \chi_j$ denote the right-hand partial derivative
of $\chi_j$ with respect to its first argument, let $D^+_\lambda \chi_j$ denote the right-hand partial
derivative of $\chi_j$ with respect to its second argument, and let
 $D^-_\lambda \chi_j$ denote the left-hand partial
derivative of $\chi_j$ with respect to its second argument.
Then over the region
$\{ (\theta, \lambda_j) : \theta \geq  0,    0 \leq \lambda_j < \eta_j + \theta\},$
\begin{eqnarray}
D^+_\theta \chi_j (\theta,\lambda_j) & = &   |\{ k : \triangle_{j,k} > \eta_j  + \theta \}|   \label{eq.Dtheta} \\
D^+_\lambda \chi_j (\theta,\lambda_j)  & = &    |\{ k : \triangle_{j,k} > \lambda_j  \}|   \label{eq.Dlambda}
\end{eqnarray}
and over the region
$\{ (\theta, \lambda_j) : \theta \geq 0,    0 \leq  \lambda_j < \eta_j + \theta\},$
\begin{eqnarray}
D^-_\lambda \chi_j (\theta,\lambda_j)  & = &    |\{ k : \triangle_{j,k} \geq \lambda_j  \}| .  \label{eq.DlambdaLeft}
\end{eqnarray}

For $j \in J(\theta),$  \eqref{eq.cond4} holds with equality at all $\theta'$ with $\theta' > \theta$ and $\theta'-\theta$
sufficiently small,
so differentiating each side of  \eqref{eq.cond4} yields
\begin{equation}  \label{eq.dsigma_temp}
\frac{d^+ \sigma_{\theta}}{d\theta}  =  D^+_\theta \chi_j (\theta,\lambda_{j,\theta}) 
+  D^{\pm}_\lambda \chi_j (\theta,\lambda_{j,\theta}) \frac{d^+\lambda_{i,\theta}}{d\theta}
\end{equation}
where
$$
D^{\pm}_\theta \chi_j (\theta,\lambda_{j,\theta})   =\left\{ \begin{array}{cl}
D^+_\theta \chi_j (\theta,\lambda_{j,\theta})   & \mbox{if}~ \frac{d^+\lambda_{i,\theta}}{d\theta}  \geq 0   \\
D^-_\theta \chi_j (\theta,\lambda_{j,\theta})   & \mbox{if}~ \frac{d^+\lambda_{i,\theta}}{d\theta}  <  0   .
\end{array} \right.
$$
Let ${\cal E}_j$ be the finite set defined by
$$
{\cal E}_j   =\left \{ \theta \geq 0 :   \frac{d^+\lambda_{i,\theta}}{d\theta}  <  0 ~\mbox{and}~ \lambda_{j,\theta} \in 
 \{ \triangle_{j,1}, \ldots , \triangle_{j,n_j}\}\right\}. 
 $$
 Since $D^+_\lambda \chi_j (\theta,\lambda_j)  = D^-_\lambda \chi_j (\theta,\lambda_j) $
 unless $\lambda_{j,\theta} \in  \{ \triangle_{j,1}, \ldots , \triangle_{j,n_j}\},$
 $$
 D^{\pm}_\lambda \chi_j (\theta,\lambda_j)  \frac{d^+\lambda_{i,\theta}}{d\theta}  
  = D^+_\lambda \chi_j (\theta,\lambda_j)  \frac{d^+\lambda_{i,\theta}}{d\theta} 
~~~~  \mbox{if}~\theta \not\in {\cal E}_j.
 $$
 Hence, \eqref{eq.dsigma_temp} implies
 \begin{equation}  \label{eq.dsigma}
\frac{d^+ \sigma_{\theta}}{d\theta}  =  D^+_\theta \chi_j (\theta,\lambda_{j,\theta}) 
+  D^+_\lambda \chi_j (\theta,\lambda_{j,\theta}) \frac{d^+\lambda_{i,\theta}}{d\theta}
~~~~  \mbox{if}~\theta \not\in {\cal E}_j.
\end{equation}
 
Dividing each side of \eqref{eq.dsigma} by the nonzero quantity $D^+_\lambda \chi_j (\theta,\lambda_{j,\theta})$
and rearranging terms yields
\begin{equation}  \label{eq.dsigma2}
\frac{   \frac{d^+\sigma_{\theta}}{d\theta}    -     D^+_\theta \chi_j (\theta,\lambda_{j,\theta})     }
{  D^+_\lambda \chi_j (\theta,\lambda_{j,\theta})   }  = \frac{d^+ \lambda_{j,\theta}}{d\theta}
~~~~  \mbox{if}~\theta \not\in {\cal E}_j.
\end{equation}
Letting ${\cal E}=\cup_j {\cal E}_j$ and
summing each side of \eqref{eq.dsigma2} over $j\in J(\theta)$ yields
\begin{eqnarray}
\lefteqn{   \frac{d^+\sigma_{\theta}}{d\theta}   \left\{  \sum_{j \in J(\theta) } \frac{1}
  {  D^+_\lambda \chi_j (\theta,\lambda_{j,\theta})   }  \nonumber   \right\}    \nonumber } \\
  &-&
   \sum_{j \in J(\theta) } 
    \frac{  D^+_\theta \chi_j (\theta,\lambda_{j,\theta})  }
   {  D^+_\lambda \chi_j (\theta,\lambda_{j,\theta})   }    \nonumber   \\
& = &   \frac{d^+\sigma_{\theta}}{d\theta}
~~~~  \mbox{if}~\theta \not\in {\cal E}.     \label{eq.dsigma3}
\end{eqnarray}
Solving \eqref{eq.dsigma3} for $ \frac{d^+ \sigma_{\theta}}{d\theta} $  and using the expressions
\eqref{eq.Dtheta} and \eqref{eq.Dlambda} for the derivatives yields \eqref{eq.Dsigma} for $\theta \geq 0$
with $\theta \not\in {\cal E}.$
Lemma \ref{lemma.Dsigma} is proved.
\end{proof}

Lemma \ref{lemma.Dsigma} is used to establish the following lemma.   Since $\widetilde{p}_{\theta,0}=v_0+\sigma_{\theta},$
the condition  $\frac{d^+ \widetilde{p}_{0,\theta} }{d \theta } \leq 1$ for all $\theta >  0$ such that $\widetilde{p}_{0,\theta} \leq v_0 + \theta,$
is equivalent to  $\frac{d^+ \sigma_{\theta}}{d\theta} \leq 1$ for all  $\theta >  0$ such that $\sigma_{\theta} \leq  \theta. $
Therefore, the following lemma completes the proof of Proposition~\ref{prop.marg_core}
.
\begin{lemma}  \label{lemma.sigma_prime_bnd}
If  $\theta > 0$ and $\sigma_{\theta}  \leq \theta,$   then $\frac{d^+\sigma_{\theta}}{d\theta} \leq 1.$
\end{lemma}
We remark that the condition $\theta > 0$ is needed in Lemma \ref{lemma.sigma_prime_bnd}.
\begin{proof}
Since $\sigma_{\theta}$ is piecewise linear, by Lemma  \ref{lemma.Dsigma}, it is sufficient to prove that
the right-hand side of  \eqref{eq.Dsigma} is less than or equal to  one for all $\theta > 0$ such that $\sigma_{\theta}\leq \theta.$
Using reasoning similar to that behind the Markov inequality of probability theory,
inspection of Figure \ref{fig.chi_area} yields that for any $j\in J$ and $\theta > 0,$
\begin{equation}  \label{eq.iiiccc}
  |\{ k : \triangle_{j,k} > \eta_j  + \theta \}|    \leq \frac{  \chi_j(\theta, \lambda_{j,\theta }) }{\eta_j + \theta  - \lambda_{j,\theta} }.
\end{equation}
Using \eqref{eq.cond4}  (which is equivalent to $\sigma_{\theta} + \eta_j \geq  \chi_j(\theta, \lambda_{j,\theta })$) in \eqref{eq.iiiccc} and then invoking the assumption that $\sigma_\theta \leq  \theta$ yields
\begin{align}
\left| \{ k : \triangle_{j,k} > \eta_j  + \theta \} \right| 
& \leq \frac{\sigma_\theta + \eta_j }{\eta_j + \theta  - \lambda_{j,\theta}}  \nonumber \\
& \leq \frac{\theta + \eta_j  }{\theta  + \eta_j - \lambda_{j,\theta} } \quad(\mbox{if}~\sigma_\theta \leq \theta). \label{eq.num_bd}
\end{align}

Let $\widehat{J}(\theta) = J(\theta) \cap \{  j : |\{ k : \triangle_{j,k}  > \eta_j + \theta \}|\geq 1 \}.$
Assuming that  $\sigma_\theta \leq \theta,$  we apply \eqref{eq.num_bd}, the fact $\eta_j\geq 0$ for all
$j,$ and the assumption  $\sigma_\theta \leq \theta$ once again to derive:
\begin{eqnarray}
\lefteqn{   \frac{    |\widehat{J}(\theta)|    
       }{1 + 
      \sum_{j \in \widehat{J}(\theta) }   \frac{  1  }{     |\{ k : \triangle_{j,k} > \eta_j + \theta  \}|      } 
       }   \nonumber  } \\
   & \leq &  
    \frac{    |\widehat{J}(\theta)|    
      }{1 +     \sum_{j \in \widehat{J}(\theta) }   \frac{\theta  + \eta_j - \lambda_{j,\theta} }{\theta + \eta_j  } 
      }   \nonumber  \\
   & = & 
      \frac{    |\widehat{J}(\theta)|    
      }{1 +   |\widehat{J}(\theta)|    -     \sum_{j \in \widehat{J}(\theta) }   \frac{ \lambda_{j,\theta} }{\theta + \eta_j  } 
      }    \nonumber      \\   
   & \leq &
         \frac{    |\widehat{J}(\theta)|    
      }{1 +   |\widehat{J}(\theta)|    -     \sum_{j \in \widehat{J}(\theta) }   \frac{ \lambda_{j,\theta} }{\theta } 
      }  \nonumber   \\
   & \leq  &
           \frac{    |\widehat{J}(\theta)|    
      }{1 +   |\widehat{J}(\theta)|    -    \frac{ \sigma_{\theta}  }{\theta } 
      }  \nonumber    \\
  & \leq &  1   ~~~(\mbox{if}~\sigma_{\theta} \leq \theta  ).  \label{eq.Jhat_bnd}
 \end{eqnarray}
 Now \eqref{eq.Jhat_bnd} can be rewritten as
 \begin{equation}  \label{eq.getting_there}
  \frac{  \sum_{j\in  \widehat{J}(\theta) }    1  
       }{1 + 
      \sum_{j \in \widehat{J}(\theta) }   \frac{  1  }{     |\{ k : \triangle_{j,k} > \eta_j + \theta  \}|      } 
       } \leq 1   ~~~(\mbox{if}~\sigma_{\theta} \leq \theta  ).
\end{equation}
If a ratio of sums of positive numbers is less than one, and if some number in the numerator and
a number in the denominator that is smaller than the one in the numerator are both multiplied by
the same factor that is less than one, then the new resulting ratio of sums  is still less than one.
Specifically in this case, take a unit term in the numerator of the ratio in \eqref{eq.getting_there}
and a term $ \frac{  1  }{     |\{ k : \triangle_{j,k} > \eta_j + \theta  \}|      } $ from the denominator
(note that it is less than or equal to one)  and multiply them both by the  ratio  
$  \frac{   |\{ k : \triangle_{j,k} > \eta_j + \theta  \}|     }{   |\{ k : \triangle_{j,k} > \lambda_{\theta,j}  \}|   },$
which is less than or equal to one. Repeat for all $j \in \wh{J}(\theta)$.
The result is:
\begin{equation}   \label{eq.getting_there2}
 \frac{      
    \sum_{j \in \widehat{J}(\theta) }   \frac{   |\{ k : \triangle_{j,k} > \eta_j  + \theta \}|      }{     |\{ k : \triangle_{j,k} > \lambda_j  \}|      }
       }{1 + 
      \sum_{j \in \widehat{J}(\theta) }   \frac{  1  }{     |\{ k : \triangle_{j,k} > \lambda_j  \}|      } 
       }   \leq 1   ~~~(\mbox{if}~\sigma_{\theta} \leq \theta  ).
\end{equation}
Finally, if $\widehat{J}(\theta)$ is replaced by $J(\theta)$ in \eqref{eq.getting_there2}, the numerator
of the ratio is not changed because the additional terms are all zero, whereas the denominator does
not decrease.    In view of \eqref{eq.Dsigma}, this completes the proof of  Lemma \ref{lemma.sigma_prime_bnd}.
\end{proof}

\subsection{Minimum revenue core and the MID} \label{sec:mrc}
It is shown in \cite{DayMilgrom08} that among all core-payment vectors, ones that minimize the seller's revenue -- referred to as the \textit{minimum revenue core} (MRC) vectors -- minimize the sum over the buyers of each buyer's maximum possible gain from unilaterally deviating from bidding his actual value. In general, MRC is a set of payment vectors and not necessarily a unique point. A variant of the quadratic payment rule is selecting the payment vector from MRC which is nearest to the Vikrey payment vector under Euclidean distance \cite{DayCramton11}. We show in this section that our results on the MID for the quadratic payment rule also apply to its MRC variant.

The MRC is the subset of the core with the minimum sum of prices, $\mb{p}^T \mb{1},$ where $\mb{1}$ is the vector of all ones. We shall consider the unique price vector resulting when the QUADRATIC selection rule is used to select a price vector from the MRC,  with reference price vector equal to the Vikrey price vector $\mb{v}_{\theta}.$ The star network setting from Section \ref{sec:qpr} is used. Let $\mb{p}^{MRC}_{\theta}$ denote the resulting price vector in the MRC, for bid $b_0 = v_0 + \theta$ by buyer zero, for $\theta \geq 0.$

Proposition \ref{prop.marg_MRcore} below extends Proposition \ref{prop.marg_core} to the case where the payment vector is taken to be the point in MRC nearest to the Vikrey payment vector. 

\begin{prop}   \label{prop.marg_MRcore}
The price for buyer zero,  $p_{0,\theta}^{MRC},$  is piecewise linear in $\theta$ with slope less than or equal to one for all $\theta \geq 0.$
\end{prop}

In the remainder of this section, a superscript $C$ on a variable denotes that it is a variable
defined in the analysis of projection onto the core, as opposed to projection onto the MRC.
The proposition is proved below by deriving an expression for $p_{0,\theta}^{MRC}$
in terms of $p_{0,\theta}^{C},$  and applying Proposition \ref{prop.marg_core}.
Recall that for $\theta$ (equivalently, $b_0$) fixed, the core constraints are given
by \eqref{eq.sum_cons}-\eqref{eq.p0_constraint}.
Since $p_0$ appears in the sum constraint \eqref{eq.sum_cons} for each $j$, the following is true.
If  a vector $\mb{p}$ is in the core with $p_0<b_0$, and two or more of the sum constraints hold with
strict inequality, say for $j$ and $j'$,  and if for some $k$ and $k'$, 
$p_{j,k} > \underline{b}_{j,k}$ and  $p_{j',k'} > \underline{b}_{j',k'},$  there
is another vector in the core with  smaller revenue.   Such a vector can be obtained by
increasing $p_0$ by some  sufficiently small $\epsilon > 0$ and decreasing both
$p_{j,k}$ and $p_{j',k'}$ by the same $\epsilon.$

Therefore, the MRC is the subset of the core such that either $p_0=b_0,$  or
$p_0$ is so large that at most one of the sum constraints is violated if
$p_{j,k}= \underline{b}_{j,k}$ for all $(j,k).$     To put this another way, for
$1\leq j \leq J,$  let $V_{j}  = \max\{   C_j - \sum_{k} \underline{b}_{j,k}   , 0 \}$,
which is the smallest nonnegative value such that if $p_0 \geq V_{j},$
then the $j^{th}$ sum constraint for the core is satisfied if $p_{j,k}=\underline{b}_{j,k}$
for $1\leq k \leq n_j.$
Let $j^*$ denote a value of $j$ that maximizes $V_j,$  and let $V_{[2]}=\max\{V_j:j\neq j^*\}$ 
(set $V_{[2]}=-\infty$ if $J=1$).
Then the MRC is precisely the subset of the core satisfying the additional constraint
$
p_0 \geq \min\{b_0,  V_{[2]} \}.
$
That is, $p_0$ either is equal to $b_0$  (which can't be exceeded due to the IR constraint
of buyer zero)  or is so large that at most one of the sum constraints is violated when all buyers $(j,k)$ bid
$\underline{b}_{j,k}.$   Since it is also a requirement that $p_0\leq b_0,$   it follows that
$p_{0}^{MRC}=b_0$ for  $v_0  \leq b_0  \leq V_{[2]}.$

To cover the remaining possibility, suppose  $\theta$ is such that $b_0 > V_{[2]}.$
 Then for any vector $\mb{p}$ in the MRC,
 $p_0 \geq V_{[2]},$   and therefore $p_{j,k,\theta}= v_{j,k,\theta}=  \underline{b}_{j,k,\theta}$
for all  $(j,k)$ with $j\neq j^*.$   The only coordinates of $\mb{p}_{\theta}$
that remain to be determined are $p_{0,\theta}$ and $p_{j^*,k,\theta}$ for $1\leq k \leq n_{j^*}.$
This reduces to a projection of the $n_{j^*}+1$ dimensional reference
vector $(v_0, (v_{j^*,k,\theta} : 1\leq k \leq n_{j^*}))$ onto the set of
vectors $(p_0, (p_{j^*,k} : 1\leq k \leq n_{j^*}))$ such that $p_0 \in [\underline{b}_0.b_0],$
$p_{j^*,k} \in [\underline{b}_{j^*,k},b_{j^*,k}]$ and
$C_{j^*} \leq p_0 + \sum_{k=1}^{n_{j^*}} p_{j^*,k}.$ 
Let $\nu$ denote the Lagrange multiplier for this sum constraint.
Proceeding as in the previous section yields that
$p_{0,\theta}^{MRC}= \min\{ \max\{v_0+\nu, V_{[2]} \},  b_0\},$ where $\nu$ is determined by the
conditions
\begin{eqnarray}   \label{eq.exhibitA1}
\nu +\eta_{j^*}  \geq  \sum_{k=1}^{n_{j^*}}   \max\{  \min \{\eta_{j^*}+\theta , \triangle_{j^*,k} \}
- \nu  ,   0  \},  \\
\mbox{with equality if}~\nu> 0.   \label{eq.exhibitA2}
\end{eqnarray}
Writing $\sigma_{\theta}^C$ for the variable $\sigma_{\theta}$ in the previous section, and applying the
fact $\lambda_{j^*}^C \leq \sigma_\theta^C$ to \eqref{eq.cond4}  yields
\begin{eqnarray}   \label{eq.exhibitb}
\sigma_{\theta}^C  +\eta_{j^*}  \geq  \sum_{k=1}^{n_{j^*}}   \max\{  \min \{\eta_{j^*}+\theta , \triangle_{j^*,k} \}
- \sigma_{\theta}^C    ,   0  \}   \label{eq.exhibitB}
\end{eqnarray}
Comparing \eqref{eq.exhibitA1} and  \eqref{eq.exhibitA2} to \eqref{eq.exhibitB} shows that
$\nu \leq \sigma_{\theta}^C.$    Thus, on one hand,
if $p_{0,\theta}^C \leq V_{[2]}$, then $v_0+\nu \leq v_o+\sigma_{\theta}^C  \leq V_{[2]},$
so $p_{0,\theta}^{MRC}=V_{[2]}=\max\{p_{0,\theta}^C, V_{[2]}\}.$
On the other hand, if  $p_{0,\theta}^C >  V_{[2]}$,  then $v_0 + \sigma_{\theta}^C  > V_{[2]}$
so for any $j$ the left-hand side of \eqref{eq.cond4} is strictly greater than $V_{[2]}-v_0 + \eta_j ,$
which if $j\neq j^*$ is greater than or equal to $\sum_{k=1}^{n_{j^*}} \triangle_{j,k},$  which is greater
than or equal to the right-hand side of  \eqref{eq.cond4}.  Hence, if $j\neq j^*,$  equality cannot hold
in  \eqref{eq.cond4} and $\lambda_{\theta,j}^C=0$ for $j\neq j^*.$   Therefore, 
$\sigma_{\theta}^C$ is equal to $\lambda_{j^*,\theta}^C$, and it satisfies the same conditions
\eqref{eq.exhibitA1} and  \eqref{eq.exhibitA2} as $\eta$ and is hence equal to $\eta$,
implying that $p_{0,\theta}^{MRC}=p_{0,\theta}^C=\max\{p_{0,\theta}^C, V_{[2]}\}.$  Thus,
whenever  $b_0 > V_{[2]},$  $p_{0,\theta}^{MRC}=\max\{p_{0,\theta}^C, V_{[2]}\}.$

The previous two paragraphs imply the following:
\begin{lemma}   \label{lemma.marg_MRcore}
Let $p_{0,\theta}^C$ denote the price for buyer zero for projection of $\mb{v}_{\theta}$  onto the core, and
$p_{0,\theta}^{MRC}$ denote the price for buyer zero for projection of $\mb{v}_{\theta}$   onto the MRC.
Then
$$
p_{0,\theta}^{MRC}  =\left\{ \begin{array}{cl}
 b_0 & \mbox{if}~v_0 \leq b_0 \leq  V_{[2]}    \\
 \max\{  p_{0,\theta}^{C}, V_{[2]}  \}  &   \mbox{if}~b_0 \geq  V_{[2] }
\end{array} \right.
$$
\end{lemma}
Proposition   \ref{prop.marg_MRcore} is a corollary of Lemma   \ref{lemma.marg_MRcore}
and Proposition   \ref{prop.marg_core}.

\section{The MID for General Core Selecting Payment Rules} \label{sec:mid-general}
We now obtain a lower bound on the worst case MID for any core-selecting payment rule, where the worst case is over
market environments.  It is assumed that the winner determination rule is efficient. 
This bound quantifies the loss in the incentives for truthful bidding if core-selecting outcome is imposed as a constraint.
Notice that the result applies to any core-selecting payment rule, and not necessarily the quadratic payment rule, and
not necessarily a payment rule involving a reference price vector such as the Vikrey price vector.

\begin{prop} \label{prop.mid-lowerbd} Let $w$ be an integer with $w \geq 2.$
For any core-selecting payment rule, the worst case MID, over all scenarios for which there are $w$ winners, is at least $1-\frac{1}{w}.$
\end{prop}

\begin{proof}
Select $\delta$ with $0 < \delta \leq \frac{1}{w-1}.$   Consider an auction of 
$w$ items and $w+1$ buyers (i.e., $|M|=w$ and $|N|=w+1$), as follows.  There are $w$ small buyers, with each interested in a distinct item,
and one large buyer, interested in all items.  Consider two scenarios, presented in reverse order. In scenario two, the small buyers each bid $1+\delta;$ the large buyer bids $w;$ the small buyers win. In order for the price vector of the winners  $(p_1, \ldots  , p_w)$ to be in the core, it is necessary that the sum of the prices be at least $w.$ Thus, for some $i^*,$  $p_{i^*}\geq 1.$
In scenario one, suppose buyer $i^*$ bids $1-(N-1)\delta,$  and the bids of the other buyers are the same as in scenario two, i.e. the other small buyers bid $1+\delta$ and the large buyer bids $w.$ The sum of bids of the small buyers is $w$, equal to the bid of the large buyer. Suppose the tie is broken in favor of the small buyers. Due to the IR constraints and the requirement of price vector being in the core of reported bids, the price vector in scenario one must equal the bid vector. In particular, the price paid by buyer $i^*$ in scenario one is $1-(w-1)\delta.$

Observe that in progressing from scenario one to scenario two, buyer $i^*$ increases his bid from $1-(w-1)\delta$ to $1+\delta,$  an increase of $w\delta.$ The bids of the other buyers and the set of winners is the same in the two scenarios. The price of buyer $i^*$ increases from $1-(w-1)\delta$ to at least one, an increase of at least $(w-1)\delta.$   Thus, by increasing his already winning bid, buyer $i^*$ causes his payment to increase by at least $1-\frac{1}{w}$ times the amount of the bid increase.
\end{proof}

\section{Conclusions and Future Work} \label{sec:conclusions}
The marginal incentive to deviate (MID) is a metric to measure the incentive to deviate from truthful bidding if core-selecting outcome is imposed as a constraint in an auction. We obtain lower and upper bounds on the MID.

An immediate direction for future work is to extend our results to general combinatorial auctions. Our focus has mostly been on quadratic payment rule. Analysis of the incentive properties of some other commonly used core-selecting payment rules is needed. For many practical scenarios, strategy-proof behavior is incompatible with other design objectives. A broader research question is how to design approximate strategyproof mechanisms and compare them.


\bibliographystyle{plain}
\bibliography{/Users/brucehajek/Documents/Papers/BIBS/auctions}

\end{document}